\documentclass[letterpaper, 11pt]{article}
\usepackage[utf8]{inputenc}
\usepackage[T1]{fontenc}
\usepackage{lmodern}
\usepackage{microtype}
\usepackage{amsmath}
\usepackage{amsfonts}
\usepackage{amsthm}
\usepackage{hyperref}
\usepackage{ifthen}
\usepackage[dvipsnames]{xcolor}

\newtheorem{theorem}{Theorem}       
\newtheorem{lemma}{Lemma}

\newcommand{\dist}[2]{%
  \ifthenelse { \equal {#1} {} }%
  {\delta_B}
  {\delta_B(#1, #2)}}

\newcommand{\origAlg}{Alg_I}

\begin{document}
\null \vskip 2em%
\begin{center}
{\LARGE Trade-off between Time, Space, and Workload:\\ the case of the Self-stabilizing Unison} \vskip  1.5em%

  {\large \lineskip .5em%
    \begin{tabular}[t]{ccc}%
      Stéphane Devismes
    \end{tabular}}%
  \vskip .5em%

  \emph{\footnotesize Laboratoire MIS, Université de Picardie,\\
    33 rue Saint Leu - 80039 Amiens cedex 1, France} \vskip 1em

  {\large \lineskip .5em%
    \begin{tabular}[t]{ccc}%
      David Ilcinkas & Colette Johnen & Frédéric Mazoit
    \end{tabular}}%
  \vskip .5em%

  \emph{\footnotesize LaBRI, Universit\'e de Bordeaux, 351 cours de la
    Lib\'eration, F-33405 Talence cedex, France}

  \let\thefootnote\relax\footnotetext{Email Addresses:
    \url{stephane.devismes@u-picardie.fr} (Stéphane Devismes),
    \url{david.ilcinkas@labri.fr} (David Ilcinkas),
    \url{johnen@labri.fr} (Colette Johnen),
    \url{frederic.mazoit@labri.fr} (Frédéric Mazoit)}
\end{center}
\vskip 1.5em

\hrule
\subsection*{Abstract}

We present a self-stabilizing algorithm for the (asynchronous) unison
problem which achieves an efficient trade-off between time, workload,
and space in a weak model.
Precisely, our algorithm is defined in the atomic-state model and
works in anonymous networks in which even local ports are unlabeled.
It makes no assumption on the daemon and thus stabilizes under the
weakest one: the distributed unfair daemon.

In a $n$-node network of
diameter $D$ and assuming a period $B \geq 2D+2$, our algorithm only
requires $O(\log B)$ bits per node to achieve full polynomiality as it
stabilizes in at most $2D-2$ rounds and $O(\min(n^2B, n^3))$ moves.
In particular and to the best of our knowledge, it is the first
self-stabilizing unison for arbitrary anonymous networks achieving an
asymptotically optimal stabilization time in rounds using a bounded
memory at each node.

Finally, we show that our solution allows to efficiently simulate
synchronous self-stabilizing algorithms in an asynchronous
environment.  This provides a new state-of-the-art algorithm solving
both the leader election and the spanning tree construction problem in
any identified connected network which, to the best of our knowledge,
beat all existing solutions of the literature.

\bigskip

\hrule

\section{Introduction}

\paragraph{Context.}
\textit{Self-stabilization} is a general non-masking and lightweight
fault tolerance paradigm~\cite{Di74,AlDeDuPe19}. Precisely, a
distributed system achieving this property inherently tolerates
\textit{any} finite number of transient faults.\footnote{A {\em transient
fault} occurs at an unpredictable time, but does not result in a
permanent hardware damage. Moreover, as opposed to intermittent
faults, the frequency of transient faults is considered to be low.}
Indeed, starting from an arbitrary configuration, which may be the
result of such faults, a self-stabilizing system recovers \emph{within
finite time}, and without any external intervention, a so-called
\emph{legitimate configuration} from which it satisfies its
specification.

The difficulty of achieving fault tolerance in distributed systems
mainly relies on their asynchronous aspect. The impossibility of
achieving consensus in an asynchronous system in spite of at most one
 process crash~\cite{FLP85} is a famous example illustrating this fact.
 Thus, fault tolerance, and in particular self-stabilization, often
 requires some kind of barrier synchronization to control the
 asynchronism of the system by making processes progress roughly at
 the same speed.

 In that spirit, the \emph{asynchronous unison problem} (unison for
 short) is a basic yet fundamental problem that helps the design of
 asynchronous distributed systems, especially self-stabilizing ones.
 The unison problem consists in maintaining a local clock at each
 node; the domain of clocks being infinite or bounded. Each node
 should increment its own clock infinitely often.\footnote{In case the
 clock values are bounded, increments are modulo some value $B$,
 called the \emph{period}.} Furthermore, the safety property of the
 unison requires the difference between the clocks of any two
 neighbors to always be at most one increment. Notice that this
 problem can be trivially generalized (as done here) by conditioning
 increments at each node $p$ to the satisfaction of some local
 predicate $P(p)$ (\emph{n.b.}, we retrieve the initial problem if
 $P(p) \equiv true$).

Unison has numerous applications, especially in
self-stabilization. Among others, it can be used to simulate
synchronous systems in asynchronous environments~\cite{AD17j,DDL19j},
free an asynchronous system from its fairness assumption (using the
cross-over composition)~\cite{BeGaJo01}, facilitate the termination
detection~\cite{BlJoLePe22}, or achieve infimum computation and local
resource allocation~\cite{BP08c}.

In this paper, we consider the unison problem in the most commonly
used model of the self-stabilizing area: the \emph{atomic-state}
model~\cite{Di74,AlDeDuPe19}.  This model is a locally-shared memory
model with composite atomicity: the state of each node is stored into
registers and these registers can be directly read by neighboring
nodes; moreover, in one atomic step, a node can read its state and
that of its neighbors, perform some local computations, and update its
state accordingly.
In the atomic-state model, asynchrony is materialized by an adversary
called \emph{daemon} that restricts the set of possible executions.
We consider here the weakest (\emph{i.e.}, the most general) daemon:
the \emph{distributed unfair daemon}.

Self-stabilizing algorithms are mainly compared according to their
\emph{stabilization time}, \emph{i.e.}, the worst-case time to reach a
legitimate configuration starting from an arbitrary one. In the
atomic-state model, stabilization time can be evaluated in terms of
{\em rounds} and {\em moves}.  Rounds~\cite{CDPV02c} capture the
execution time according to the speed of the slowest nodes. Moves
count the number of local state updates.  So, the move complexity is
rather a measure of work than a measure of time.

It turns out that obtaining efficient stabilization time both in
rounds and steps is a difficult issue. Usually, techniques to design
an algorithm achieving a stabilization time polynomial in moves
usually makes its rounds complexity inherently linear in $n$, the
number of nodes; see,
e.g.,~\cite{CoDeVi09,AlCoDeDuPe17,DeJo19,DeIlJo22}. Conversely, achieving the
asymptotic optimality in rounds, \emph{i.e.}, $O(D)$ where $D$ is the
network diameter, commonly makes the stabilization time in
moves exponential; see, \emph{e.g.},~\cite{DeJo16,GHIJ19}.
In a best-effort spirit, Cournier
\emph{et al.}~\cite{CoRoVi19} have proposed to study what they call
\emph{fully-polynomial} self-stabilizing solutions, \emph{i.e.},
self-stabilizing algorithms whose round complexity is polynomial on the
network diameter and move complexity is polynomial on the network
size.\footnote{Actually, in~\cite{CoRoVi19}, authors consider atomic
steps instead of moves. However, these two time units essentially
measure the same thing: the workload. By the way, the number of moves
and the number of atomic steps are closely related: if an execution
$e$ contains $x$ steps, then the number $y$ of moves in $e$ satisfies
$x \leq y \leq n\cdot x$.}

\paragraph{Contribution.}

We propose the first fully-polynomial self-stabilizing unison in the
atomic-state model assuming a distributed unfair daemon. This
algorithm works in an anonymous network of arbitrary
topology. Moreover, it does not require any local port labeling at
nodes. In that sense, the computational model we use is close to the
\emph{stone age} model of Emek and Wattenhofer~\cite{EmWa13}.

To the best our our knowledge, this is a first fully-polynomial
self-stabilizing algorithm solving a \emph{dynamic
  problem}.\footnote{As opposed to a \emph{static problem} that
defines a task of calculating a function that depends on the system in
which it is evaluated~\cite{Ti06}.} This is also the first
self-stabilizing unison for arbitrary anonymous networks achieving an
asymptotically optimal stabilization time in rounds (\emph{i.e.},
$O(D)$) using a bounded memory at each node.

In more detail, assuming a period $B \geq 2D+2$, our solution
stabilizes in at most $2D-2$ rounds and $O(\min(n^2B, n^3))$ moves
using $O(\log B)$ bits per node.
Overall, our unison achieves an outstanding trade-off
between time, workload, and space.

We also analyze the efficiency of our algorithm to simulate any
synchronous self-stabilizing algorithm in an asynchronous environment
(under the unfair daemon).  If the input synchronous self-stabilizing
algorithm is \emph{silent}\footnote{In the atomic-state model, a
self-stabilizing algorithm is silent if all its executions terminate.}
and stabilizes in at most $T$ synchronous rounds, then its simulation is also
silent and self-stabilizing; moreover, its stabilization time is at most
$5D+3T$ rounds and $O(\min(n^2B, n^3))+nT$ moves using $O(M+\log(B))$
bits per node, where $M$ is the memory requirement of the input
algorithm.

An important consequence of this latter result is that one can easily
obtain the state-of-the-art leader election and BFS spanning tree
construction of the literature for asynchronous identified and
arbitrary connected networks simply by simulating the synchronous
algorithm of Kravchik and Kutten~\cite{KrKU13}. Precisely, by
simulating this algorithm using our unison, we obtain a stabilization
time in $O(D)$ rounds and $O(\min(n^2B, n^3))$ moves using
$O(\log(N))$ bits per node, where $N$ is any upper bound on $n$.  To
the best of our knowledge, there was no such an efficient solution
until now in the literature.

\paragraph{Related Work.}

The asynchronous unison studied here is a variant of the {\em
  synchronous unison} problem proposed by Even and
Rajsbaum~\cite{ER90}. This latter problem is dedicated to synchronous
systems and requires all clocks increment infinitely often and become
eventually fully synchronized.  In~\cite{ER90}, Even and Rajsbaum
consider this problem in a non-fault-tolerant context, yet assuming
that nodes do not necessarily start at the same time.

Gouda and Herman~\cite{GH90j} have proposed the first self-stabilizing
synchronous unison. Their algorithm works in anonymous synchronous
systems of arbitrary connected topology using infinite clocks.  A
solution working with the same settings, yet implementing bounded
clocks, is proposed in~\cite{ADG91j}.

Johnen \emph{et al.} investigated  the asynchronous self-stabilizing 
unison in oriented trees in~\cite{JADT02j}.
The first self-stabilizing asynchronous unison for general graphs was
proposed by Couvreur \emph{et al.}~\cite{CoFrGo92} in the link-register
model (a locally-shared memory model without composite
atomicity). However, no complexity analysis was given. Another
solution which stabilizes in $O(n)$ rounds has been proposed by
Boulinier \emph{et al.}~\cite{BoPeVi04} in the atomic-state model
assuming a distributed unfair daemon. Its move complexity is shown
in~\cite{DePe12} to be in $O(Dn^3+\alpha n^2)$, where $\alpha$ is a
parameter of the algorithm that should satisfies $\alpha \geq L - 2$,
where $L$ is the length of the longest hole in the network.
Boulinier proposes in his PhD thesis a parametric solution which
generalizes both the solutions of
\cite{CoFrGo92} and \cite{BoPeVi04}. In particular, the complexity
analysis of this latter algorithm reveals an upper bound in
$O(D.n)$ rounds on the stabilization time of the atomic-state model version of the Couvreur \emph{et
al.}'s algorithm.

Awerbuch \emph{et al.}~\cite{AwKuMaPaVa93} proposes a self-stabilizing
unison (called clock synchronizer in their paper) that stabilizes in
$O(D)$ rounds using an infinite state space. The move complexity of
their solution is not analyzed.
An asynchronous self-stabilizing unison algorithm is given
in~\cite{DeJo19}. It stabilizes in $O(n)$ rounds and $O(\Delta.n^2)$
moves using unbounded local memories.
Emek and Keren present in the stone age model~\cite{EmKe21} a
self-stabilizing unison that stabilizes in $O(B^3)$ rounds, where $B$
is an upper bound on $D$ known by all nodes. Their solution requires
$O(\log(B))$ bits per nodes. Moreover, since node activations are
assumed to be fair, the move complexity of their solution cannot be
bounded.

In~\cite{DIJM23}, we propose an algorithm that transforms any
terminating synchronous algorithms into an asynchronous silent
self-stabilizing fully-polynomial algorithm.  The memory requirement
of the produced algorithm is in $O(T\times M)$ bits per nodes, where
$T$ and $M$ are the time and space complexities of the input
algorithm.  This transformer thus cannot practically build solutions
for dynamic problems such as unison.  Moreover, although it works on a
strictly smaller class of algorithms, the synchronizer of the current
paper has similar round and move complexities as the transformer
of~\cite{DIJM23} while achieving a much better memory requirement.

\paragraph{Roadmap.}
The rest of the paper is organized as follows. The next section is
dedicated to the computational model and basic definitions. In
Section~\ref{sect:algo}, we present our unison algorithm, prove its
self-stabilization, and study its time complexity. In
Section~\ref{sec:synchroniseur} deals with the
simulation of synchronous self-stabilizing algorithms in an
asynchronous environment using our unison algorithm.

\section{Preliminaries}

\subsection{Networks}

We consider \emph{distributed systems} made of $n \geq 1$
interconnected nodes. Each node can directly communicate through
channels with a subset of other nodes, called its neighbors.  We
assume that the network is connected and that communication is
bidirectional.

More formally, we model the topology by a connected simple graph
$G = (V, E)$, where $V$ is the set of \emph{nodes} and $E$ is the set
of \emph{edges}.  If $\{p, q\}$ is an edge, then $q$ is a
\emph{neighbor} of $p$.  We denote by $N(p)$ the set of neighbors of
$p$.

A \emph{path} is a finite sequence $P=p_0p_1\cdots p_l$ of nodes such
that consecutive nodes in $P$ are neighbors. We say that $P$ is
\emph{from} $p_0$ \emph{to} $p_l$.  The \emph{length} of the path $P$
is the number $l$.  Since we assume that $G$ is \emph{connected}, then
for every pair of nodes $p$ and $q$, there exists a path from $p$ to
$q$.  We can thus define the \emph{distance} between two nodes $p$ and
$q$ to be the minimum length of a path from $p$ to $q$.  The
\emph{diameter} $D$ of $G$ is the maximum distance between nodes of
$G$.

\subsection{Computational Model: the Atomic-state Model}

Our algorithm runs on a variant of the \emph{atomic-state
  model}~\cite{AlDeDuPe19} in which nodes communicate using a finite
number of locally shared registers, called \emph{variables}.  The
\emph{state} of a node is defined by the values of its local
variables.  A \emph{configuration} of the system is a vector
consisting of the states of each node.

In one indivisible move, a node $p$ reads its own variables and the
set of states of its neighbors.  Our algorithm is described as a
finite set of \emph{rules} of the form $label:guard \to\ action$.
Labels are only used to identify rules in the reasoning.  A
  \emph{guard} is a Boolean predicate involving the state of the node
  and the set of states of its neighbors.  The \emph{action} part of
a rule updates the state of the node.  A rule can be executed only if
its guard evaluates to \emph{true}; in this case, the rule is said to
be \emph{enabled}.  By extension, a node is said to be enabled if at
least one of its rules is enabled.  We denote by $Enabled(\gamma)$ the
subset of nodes that are enabled in configuration $\gamma$.

In the model, executions proceed as follows. Given a configuration
$\gamma$ with $Enabled(\gamma)\neq\emptyset$, a so-called
\emph{daemon} selects a nonempty set
$\mathcal X \subseteq Enabled(\gamma)$; then every node of
$\mathcal X$ \emph{atomically} executes one of its enabled rules,
leading to a new configuration $\gamma'$.  The atomic transition from
$\gamma$ to $\gamma^\prime$ is called a \emph{step}. We also say that
each node of $\mathcal X$ executes an \emph{action} or simply a
\emph{move} during the step from $\gamma$ to $\gamma^\prime$.  The
possible steps induce a binary relation over $\mathcal C$, denoted by
$\mapsto$.  An \emph{execution\/} is a maximal sequence of
configurations $e=\gamma^0\gamma^1\cdots \gamma^i \cdots$ such that
$\gamma^{i-1}\mapsto\gamma^i$ for all $i>0$.  The term ``maximal''
means that the execution is either infinite, or ends at a
\emph{terminal} configuration $\gamma^f$ with
$Enabled(\gamma^f)=\emptyset$.  An algorithm which does not admit any infinite
execution is called \emph{silent}.

As explained before, each step from a configuration to another is
driven by a daemon.  We define a daemon as a predicate over
executions.  We say that an execution $e$ is \emph{an execution under
  the daemon $S$} if $S(e)$ holds.  In this paper we assume
  that the daemon is \emph{distributed} and \emph{unfair}, meaning
  that it has no constraints, except that at each step it must select a nonempty
  set of enabled nodes.  It might, for example, never select a specific
  enabled node unless it is the only enabled node.

We use two units of measurement to evaluate the time complexity:
\emph{ moves} and \emph{rounds}.
The definition of a round uses the concept of \emph{neutralization}: a
node $p$ is \emph{neutralized} during a step
$\gamma^i \mapsto \gamma^{i+1}$, if $p$ is enabled in $\gamma^i$ but
not in configuration $\gamma^{i+1}$, and does not execute any action
in the step $\gamma^i \mapsto \gamma^{i+1}$.
Then, the rounds are inductively defined as follows.  The first round
of an execution $e = \gamma^0\gamma^1\cdots$ is the minimal prefix
$e'$ such that every node that is enabled in $\gamma^0$ either
executes a rule or is neutralized during a step of $e'$. If $e'$ is
finite, then let $e''$ be the suffix of $e$ that starts from the last
configuration of $e'$; the second round of $e$ is the first round of
$e''$, and so on and so forth.

The \emph{stabilization time} of a self-stabilizing algorithm is the
maximum time (in moves or rounds) over every execution possible under
the considered daemon (starting from any initial configuration) to
reach a legitimate configuration.

\section{A unison algorithm}\label{sect:algo}
 
\subsection{The algorithm}

\paragraph{Data structures.}
  Let $B\geq 2D+2$ be an integer. Each node $p$ maintains a single
  variable $p.v$ of datatype $Pairs = \{(C,x)\ |\ x \in [-B, B[\} \cup
      \{(E,x)\ |\ x \in [-B, 0[\}$. 
          In the algorithm, $p.v$ will be accessed and modified
          implicitly as follows:
          \begin{itemize}
            \item $p.s$, called the {\em status} of $p$, will denote
              the left field of the pair $p.v$,
             \item $p.c$, called the {\em clock} of $p$, will denote
               the right field of the pair $p.v$.
          \end{itemize}
         For example, if $p.v = (s,c)$, then $p.s=s$ and $p.c=c$.
         Furthermore, any assignment $p.s := s$ (resp., $p.c := c$) should
         be understood as $p.v := (s,p.c)$ (resp., $p.v := (p.s,c)$).
         Finally, a node $p$ such that $p.s = C$ is said to
         be {\em correct}; otherwise it is an {\em erroneous} node (in
         other words, a node in error).



We define the infix function $+_B$ as follows:
\begin{eqnarray*}
  B-1 +_B1 &=& 0\\
  n+_B1 &=& n+1\qquad\qquad\qquad\text{if $n\neq B-1$}\\
  n+_B (m+1) &=& (n+_Bm)+_B1
\end{eqnarray*}

We also define a distance $\dist{}{}$:
\begin{eqnarray*}
  \dist{n}{n}&=&0\\
  \dist{n}{n+_B1}&=&1\\
  \dist{n+_B1}{n}&=&1\\
  \dist{n}{m}&=&2 \qquad\text{otherwise.}
\end{eqnarray*}

If $\gamma^0\gamma^1\cdots$ is an execution, we respectively denote by
$p.s^i$ and $p.c^i$ the value of $p.s$ and $p.c$ in $\gamma^i$.

\paragraph{Some predicates.}
Although they are a bit misleading because they suggest that a node
can access its neighbors directly, we use the following notations:\label{macro}
\begin{eqnarray*}
  \texttt{Macro1}\quad \exists q\in N(p), \text{Pred} (\texttt{st}_q) &:=&
                                                      \exists \texttt{st}\in \{\texttt{st}_q\mid q\in N(p)\}, \text{Pred} (\texttt{st})\\
  \texttt{Macro2}\quad \forall q\in N(p), \text{Pred} (\texttt{st}_q) &:=&
                                                      \forall \texttt{st}\in \{\texttt{st}_q\mid q\in N(p)\}, \text{Pred} (\texttt{st})\\
\end{eqnarray*}

\begin{eqnarray*}
  root(p)
  &:=&\bigl(p.s=E
       \wedge \neg(\exists q\in N(p),\;
       q.s=E\wedge q.c<p.c)\bigr)\\
  &\vee&\bigl(p.s=C\wedge \exists q\in N(p),\; p.c<q.c\wedge
         \dist{q.c}{p.c}\geq 2\bigr)
  \\\noalign{\medskip}
  activeRoot(p)
  &:=&root(p)\wedge (p.c \neq -B\vee p.s=C)
  \\\noalign{\medskip}
  errorPropag(p, i)&:=& \exists q\in N(p),\; q.s=E \wedge q.c<i<p.c
  \\\noalign{\medskip}
  canClearE(p)&:=& p.s=E\\
  &\wedge&\forall q \in N(p),\; \bigl(q.c\in\{p.c-1, p.c, p.c+1\}\wedge\\
  &&\phantom{\forall q \in N(p),\; \bigl(}
     (q.c\neq p.c+1\vee q.s=C)\bigr)
  \\\noalign{\medskip}
  unisonMove(p)&:=& p.s=C \wedge \forall q \in N(p),\; q.c\in\{p.c, p.c+_B1\}
\end{eqnarray*}

\paragraph{The rules.}
We rarely use a unison algorithm alone.  It is merely a tool to help
another algorithm. It thus makes sense that our algorithm depends on
some properties which are external to the unison algorithm and its
variables.  Our algorithm uses a predicate $P_\text{aux}$ which is not defined.  As
a matter of fact, its influence on the analysis of the algorithm is
very limited. We will specialize this predicate in
Section~\ref{sec:synchroniseur} when using our unison algorithm as
a synchronizer.

\begin{itemize}
\item $R_R: activeRoot(p) \longrightarrow p.c := -B \; ; \; v.s:=E$
\item
  $R_P(i): errorPropag(p, i)\longrightarrow p.c := i \; ; \; p.s := E$
\item $R_C: canClearE(p) \longrightarrow p.s := C$
\item
  $R_U: unisonMove(p)\wedge\\\phantom{R_U: }\qquad (P_\text{aux}(p)\vee \exists
  q\in N(p), q.c=p.c+_B 1) \longrightarrow p.c:=p.c+_B1$
\end{itemize}
We set the following priorities:
\begin{itemize}
\item $R_R$ has the highest priority.
\item $R_P(i)$ has a higher priority than $R_P(i+l)$ for $l>0$.
\item $R_C$ and $R_U$ have the lowest priority.
\end{itemize}

A node $p$ is a \emph{root} if $root(p)$.  In the following, an
\emph{error rule} is either the rule $R_R$ or a rule $R_P(i)$.

The
\emph{legitimate} configurations are the configurations in which the
only rule which can be executed is the rule $R_U$. Another equivalent
characterization of legitimate configurations will be given in
Section~\ref{sec:legitimate}.

The following remark is quite important.  Since, when encountering an
error, the clock of a node becomes negative, and since no nodes in
error can have a non-negative clock, it is natural to expect the ``error
recovery phase'' to correspond to the time zone $[-B, 0[$, and the
    interval $[0, B[$ to correspond to the ``legitimate
        configurations''. This would suggest a round complexity of
        $\Omega(B)$. But this intuition is false. If a configuration
        $\gamma$ is such that $p.s=C$ and $p.c=-B$ for every node $p$,
        then $\gamma$ is a legitimate configuration.


\subsection{Preliminary results}
\begin{lemma}\label{lem:noRootCreation}
  Let $\gamma^a\mapsto\gamma^b$ be a step.  If $p$ is a root in
  $\gamma^b$, then it also is in $\gamma^a$.
\end{lemma}
\begin{proof}
  Suppose by contradiction that $p$ is a root in $\gamma^b$ and not a
  root in $\gamma^a$.

  We consider two cases.
  \begin{itemize}
  \item Suppose that $p.s^b=E$.  Thus there exists no $q\in N(p)$ such
    that $q.s^b=E$ and $q.c^b<p.c^b$.

    If $p.s^a=E$ and no $q\in N(p)$ is such that $q.s^a=E$ and
    $q.c^a<p.c^a$, then $p$ is a root in $\gamma^a$, a contradiction.

    We claim that in all remaining cases, $p$ executes an error rule in
    $ \gamma^a\mapsto\gamma^b$.  Indeed,
    \begin{itemize}
    \item if $p.s^a=E$ and there exists $q\in N(p)$ such that
      $q.s^a=E$ and $q.c^a<p.c^a$, then $p$ cannot execute the rule $R_U$, $q$
      cannot execute the rules $R_U$ or $R_C$, and thus $q.s^b=E$.  We have
      $q.c^b \leq q.c^a < p.c^a$. So if $p.c^b \geq p.c^a$, then $p$ is
      not a root in $\gamma^b$.  Thus $p$ must execute an error rule in
      $\gamma^a\mapsto\gamma^b$.
    \item if $p.s^a=C$, then $p$ must also execute an error rule in
      $\gamma^a\mapsto\gamma^b$.
    \end{itemize}

    Now two cases arise.
    \begin{itemize}
    \item If $p$ executes the rule $R_R$, then $p$ is a root in
      $\gamma^a$, a contradiction.
    \item If $p$ executes a rule $R_P(i)$ in $\gamma^a\mapsto\gamma^b$,
      then there exists $r\in N(p)$ such that $r.c^a=i-1$ and
      $r.s^a=E$.  But since $r.s^a=E$, $r$ cannot execute the rule
      $R_U$, and because of $p$, $r$ cannot execute the rule $R_C$.  Thus
      $r.s^b=E$ and $r.c^b<p.c^b$, which contradicts the hypothesis.
    \end{itemize}

  \item Suppose that $p.s^b = C$.  Thus, there exists $q\in N(p)$ such
    that $p.c^b<q.c^b$ and $\dist{p.c^b}{q.c^b}\geq 2$.  Note that
    this implies that $q.c^b\geq p.c^b+2$.

    Since $p$ does not execute an error rule in
    $\gamma^a \mapsto\gamma^b$, either $p.c^b=p.c^a$ or
    $p.c^b=p.c^a+_B1$.
    \begin{itemize}
    \item Suppose that $p.c^b=p.c^a$.  Let us study what happens
      during $\gamma^a\mapsto\gamma^b$.
      \begin{itemize}
      \item If $q$ executes the rule $R_R$, then $q.c^b=-B$, which
        contradicts the fact that $p.c^b< q.c^b$.
      \item If $q$ executes the rule $R_U$, then it means that
        $q.c^a \leq p.c^a$.  And since $p.s^b=C$, $p$ does not execute
        an error rule and thus $q.c^b\leq p.c^b+1$, a contradiction.
      \item If $q$ executes no rules or the rule $R_C$, then
        $q.c^a=q.c^b$, and since $q.c^a\geq p.c^a+2$, $p$ cannot execute
        the rule $R_C$. Thus, we have $p.s^a=C$, which implies that $p$ is also
        a root in $\gamma^a$, a contradiction.
      \item If $q$ executes a rule $R_P$, then
        $q.c^a>q.c^b\geq p.c^b+2=p.c^a+2$. Thus, we have $q.c^a > p.c^a+2$,
        which prevents $p$ from executing the rule $R_C$.  Thus
        $p.s^a=C$, and since $p$ is not a root in $\gamma^a$,
        $\dist{p.c^a}{q.c^a}\leq 1$, a contradiction.
      \end{itemize}

    \item Suppose that $p.c^b=p.c^a+_B1$.  Since
      $B-1\geq q.c^b\geq p.c^b+2$, $p.c^b=p.c^a+1$.  This implies that
      $p$ executes the rule $R_U$ during $\gamma^a\mapsto\gamma^b$, and
      thus $q.c^a\in\{p.c^a, p.c^a+_B1\}=\{p.c^a, p.c^a+1\}$.
      If $q$ executes the rule $R_U$ during
        $\gamma^a\mapsto\gamma^b$, then $q.c^a= p.c^a$. In this case,
        we have $q.c^b= p.c^b$, a contradiction.  Otherwise, we have
        $q.c^b \leq p.c^a + 1 = p.c^b $, again a contradiction.
    \end{itemize}
  \end{itemize}
\end{proof}

\begin{lemma}\label{lem:root_RC}
  Let $\gamma^a\mapsto\gamma^b$ be a step, and let $r$ be a root in
  $\gamma^a$ which executes the rule $R_C$ during
  $\gamma^a\mapsto\gamma^b$.  Then $r.c^a=-B$ and $r$ is not a root in
  $\gamma^b$.
\end{lemma}
\begin{proof}
  Since $R_R$ has a higher priority than $R_C$, the guard of $R_R$ is
  false at $r$ in $\gamma^a$. So, as $r$ is a root in $\gamma^a$, we
  necessarily have $r.c^a=-B$.

  Then, since $r$ executes the rule $R_C$ during $\gamma^a\mapsto\gamma^b$, we
  have $r.s^b=C$.  Moreover, to allow $r$ to execute the rule $R_C$,
  every $q\in N(r)$ should satisfy $q.c^a\leq -B+1$. Now, as
  $r.c^a=-B$, no $q\in N(r)$ with $q.c^a=-B+1$ can execute the rule $R_U$ in
  $\gamma^a\mapsto\gamma^b$.  All this implies that $r.s^b=C$, and for every $q \in N(p)$, $\dist{q.c^b}{p.c^b}\leq 1$. So, $root(r)$ is
  false in $\gamma^b$, {\em i.e.}, $r$ is not a root in $\gamma^b$.
\end{proof}

A path $P=p_0p_1\cdots p_l$ in $G$ is \emph{decreasing} in a
configuration $\gamma$ if for each $0\leq i<l$, $p_i.c>p_{i+1}.c$. Moreover,
$P$ is an \emph{$E$-path} if it is decreasing, all its nodes are in
error, and its last node is a root.

\begin{lemma}\label{lem:E-path}
  Let $\gamma$ be a configuration.  Any node $p$ in error is the first
  node of an $E$-path.
\end{lemma}
\begin{proof}
  We prove our lemma by induction on $p.c$. If $p.c=-B$, then $p$ is a
  root and $P=p_0$ satisfies the required conditions.

  Suppose that $p.c>-B$.  If $p$ is a root, then $P=p_0$ satisfies the
  required conditions.  Otherwise, there exists $q\in N(v)$ such that
  $q.c<p.c$ and $q.s=E$.  By induction, there exists an $E$-path $P'$
  starting at $q$.  We can add $p$ at the beginning of $P'$ to obtain a
  path $P$ which satisfies all required conditions.
\end{proof}

\subsection{Legitimate configurations}\label{sec:legitimate}

A configuration $\gamma$ is said to be \emph{almost clean} if
\begin{itemize}
\item every root $r$ satisfies $r.c=-B$ and $r.s=E$, and
\item every two neighbors $p$ and $q$ satisfy
  $\dist{p.c}{q.c}\leq 1$.
\end{itemize}

\begin{lemma}\label{lem:almostClean}
  A configuration is almost clean if and only if no nodes can execute an
  error rule.
\end{lemma}
\begin{proof}
  Suppose that $\gamma$ is almost clean. Since every root $r$ is such
  that $r.c=-B$ and $r.s=E$, no nodes can execute the rule $R_R$, and
  since every neighbors $p$ and $q$ are such that
  $\dist{p.c}{q.c}\leq 1$, no nodes can execute a rule $R_P$.

  Conversely, suppose that $\gamma$ is not almost clean. A root $r$
  verifying $r.c>-B$ or $r.s=C$ can execute the rule $R_R$.
    Let $p$ and $q$ be two neighbors. Assume, without loss of
    generality, that $p.c\leq q.c$. If $\dist{p.c}{q.c}\geq 2$, then
    either $p.s=E$ and $q$ can execute a rule $R_P$, or
    $p.s=C$ and $p$ can execute the rule $R_R$.
\end{proof}

\begin{lemma}\label{lem:almostCleanClosed}
  Let $\gamma^a\mapsto\gamma^b$ be a step. If $\gamma^a$ is almost
  clean, then so is $\gamma^b$.
\end{lemma}
\begin{proof}
  Assume, for the purpose of contradiction, that $\gamma^a$ is almost clean and
  $\gamma^b$ is not.

  At least one of the following two cases occurs, by
    Lemma~\ref{lem:almostClean}.
  \begin{itemize}
  \item Some root $r$ can execute the rule $R_R$ in $\gamma^b$ (i.e.,
    $r.c^b>-B$ or $r.s^b=C$).

    First, by Lemma~\ref{lem:noRootCreation}, $r$ is a root in
    $\gamma^a$, and since $\gamma^a$ is almost clean, $r.c^a=-B$ and
    $r.s^a=E$. Thus, either $r$ executes no
    rules in $\gamma^a\mapsto\gamma^b$, which is a contradiction with
    $r.c^b>-B$ or $r.s^b=C$, or $r$ executes the rule $R_C$ and $r$ is
    not a root in $\gamma^b$ by Lemma~\ref{lem:root_RC}, which also leads to a contradiction.

  \item Some node $p$ can execute a rule $R_P$ in $\gamma^b$. There
    exists $q\in N(p)$ such that $B-1 \geq p.c^b\geq q.c^b+2$ and
    $q.s^b=E$.  Since $\gamma^a$ is almost clean, no error rules are
    executed in the step $\gamma^a \mapsto \gamma^b$.  Thus $q.s^a=E$
    and $q$ executes no rules in $\gamma^a\mapsto\gamma^b$, so
    $q.c^a=q.c^b$.  Moreover, $\dist{p.c^a}{q.c^a}\leq 1$.
    This implies that $p$ must execute the rule $R_U$.  But
      then $p.c^a \geq q.c^a+1$ as
      $p.c^b = p.c^a+1 \geq q.c^b+2 = q.c^a+2 $. Thus
    $q.c^a\notin\{p.c^a, p.c^a+_B1\}$, which forbids $p$ from executing
    the rule $R_U$, a contradiction.
  \end{itemize}
\end{proof}

\begin{lemma}\label{lem:hole}
  In any almost clean configuration $\gamma$, there exists $c\in [0,B[$
  such that for any $p$, $p.c\neq c$.
\end{lemma}
\begin{proof}
  Suppose that for all $c \in [0,B[$, there exists $p$ such that
  $p.c=c$.  Hence, there is a node $p$ whose clock value is $D$
  ($p.c = D$) in $\gamma$.  We can prove by induction on $l$ that any
  node $q$ at distance at most $l$ from $p$ has a clock value in
  $[D-l, D+l]$.  We conclude that no node $p$ is such that
  $p.c=2D+1\leq B-1$, a contradiction.
\end{proof}

\begin{lemma}
\label{lem:colorValue}
  Let $\gamma$ be an almost clean configuration.  There exists
  $c_{\min}\in [-B, B[$ and $\Delta_c \leq D$ such that
  $\{p.c\mid p\in V\}=\{c_{\min}+_Bi\mid 0\leq i\leq \Delta_c\}$.
\end{lemma}
\begin{proof}
  We consider two cases.
  \begin{itemize}
  \item Suppose that there exists $p$ such that $p.c<0$.  Let
    $c_{\min}=\min(p.c\mid p\in V)$, and let $\Delta_c$ be the minimum natural
    integer such that no node $q$ is such that
    $q.c=c_{\min}+ \Delta_c+1$ ($\Delta_c$ exists by
    Lemma~\ref{lem:hole}).
  \item Suppose that no node $p$ is such that $p.c<0$.  By
    Lemma~\ref{lem:hole}, there exists $c\in [0, B[$ which is not the
    clock value of any node.  Since clock values are non-negative,
    there exists a minimum $i$ such that $c_{\min}=c+_B i$ is a clock value
    of a node $p$.  We choose $\Delta_c$ minimum such
    that no node $q$ is such that $q.c=c_{\min}+_B (\Delta_c+1)$.
  \end{itemize}
    Clearly,
    $\{c_{\min}+_B i\mid 0\leq i\leq \Delta_c\}\subseteq \{p.c\mid p\in
    V\}$.  Now, equality and the fact that $\Delta_c \leq D$ follow from the fact that $G$ is connected
    and that, between two consecutive nodes of any path, the clock
    value can only change by one.
\end{proof}

A configuration is said to be \emph{clean} if it contains no roots.
Lemma~\ref{lem:noRootCreation} implies that being clean is a closed
property.
The following
lemma gives an alternative definition of being clean, and as a
direct consequence, it implies that clean configurations are also
almost clean. It also implies that the legitimate configurations are
the clean ones.

\begin{lemma}\label{lem:clean}
  A configuration is clean if and only if nodes can only execute the
  rule $R_U$.
\end{lemma}
\begin{proof}
  Suppose that $\gamma$ is clean. Since it contains no roots, then no
  nodes can execute the rule $R_R$. Since there are no roots, then, by
  Lemma~\ref{lem:E-path}, there are no nodes in error, and thus no nodes
  can execute a rule $R_P$ or the rule $R_C$.
  
  Conversely, suppose that
  nodes can only execute the rule $R_U$. Then by
  Lemma~\ref{lem:almostClean}, $\gamma$ is almost clean. Therefore
  $\gamma$ contains no roots having the status $C$.  To prove that
  $\gamma$ does not contain any root in error, it is enough to show that
  $\gamma$ contains no nodes in error (Lemma~\ref{lem:E-path}).
  Suppose that in $\gamma$ one or several nodes are in error.
    Let $p$ be a node in error having the largest clock value.
    Since $\gamma$ is almost clean, every neighbor $q$ of $p$ satisfies
    $\dist{p.c}{q.c}\leq 1$.  By definition of $p$, a neighbor of $p$
    in error has a clock value smaller than or equal to $p.c$.  Hence, $p$ can execute
    the rule $R_C$, a contradiction.
\end{proof}

\begin{lemma}\label{lem:vivant}
  Let $e=\gamma^0\gamma^1\cdots$ be an execution such that $\gamma^0$ is
  clean.  In any configuration $\gamma^i$ of $e$, if a node $p$ satisfies
  $P_\text{aux}$, then at least one node $q$ can execute the rule $R_U$ in
  $\gamma^i\mapsto\gamma^{i+1}$.
\end{lemma}
\begin{proof}
  By Lemma ~\ref{lem:noRootCreation}, the configuration $\gamma^i$ is also clean
  (and almost clean as well by Lemma~\ref{lem:clean}).  According to Lemma~\ref{lem:colorValue}, in
  $\gamma^i$, there exists $c_{min} \in [-B, B[$ and $\Delta_c \leq D$
      such that $\{p.c\mid p\in V\}=\{c_{min} +_B l\mid 0\leq l \leq
      \Delta_c\}$. Moreover, in $\gamma^i$, the clock value of every
      neighbor of any node $p$ such that $p.c=c_{min} $ belongs to
      $\{c_{min},c_{min} +_B 1\}$.  If $\Delta_c=0$, then any node
      which satisfies $P_\text{aux}$ can execute the rule $R_U$ as all
      nodes have the same clock value and have status~$C$.  Otherwise,
      there exists a node $p$ with $p.c=c_{min} $ which has a neighbor
      $q$ such that $q.c=c_{min} +_B 1$, and so $p$ can execute the rule
      $R_U$ in~$\gamma^i$.
\end{proof}

\subsection{$D$-paths}

Recall that a path $P=p_0p_1\cdots p_l$ in $G$ is \emph{decreasing} in
a configuration $\gamma$ if for each $0\leq i<l$, $p_i.c>p_{i+1}.c$
and that $P$ is an \emph{$E$-path} if it is decreasing, all its nodes
are in error, and its last node is a root.

We extend these definitions in the following way.  A path $P$ is
\emph{gently decreasing} if, for each $0\leq i<l$, we have $p_i.c=p_{i+1}.c+ 1$.
It is a \emph{$D$-path} if it is decreasing and there exists
$0\leq j\leq l$ such that
\begin{itemize}
\item $P_C=p_0\cdots p_{j-1}$ is a (possibly empty) gently decreasing
  path of nodes in $C$,
\item $P_E=p_j\cdots p_l$ is an $E$-path.
\end{itemize}
We call $P_C$ and $P_E$ the \emph{correct} and \emph{error} parts of
$P$.

\begin{lemma}\label{lem:D-path_time}
  Let $\gamma^a\mapsto\gamma^b$ be a step, and let $P$ be a $D$-path
  in $\gamma^a$.  For any $p\in P$, node $p$ does not execute the rule $R_U$ at that step, and thus $p.c^b\leq p.c^a$.  Moreover, if
  $p\in P$ is such that $p.s^b=C$, then we have equality.
\end{lemma}
\begin{proof}
  Let $p\in P$.  Recall that $p.c$ only increases if $p$ executes the
  rule $R_U$.
  \begin{itemize}
  \item If $p$ is the last node of $P$, then in $\gamma^a$, $p$ is a
    root such that $p.s^a=E$.  Thus $p$ cannot execute the rule $R_U$ in
    $\gamma^a\mapsto\gamma^b$.
  \item If $p$ is not the last node of $P$, let $q$ be the next node
    after $p$ on $P$.  Since $P$ is decreasing in $\gamma^a$,
    $q.c^a<p.c^a$.  To be able to execute the rule $R_U$, we must have
    $q.c^a\in \{p.c^a, p.c^a+_B1\}$, which is only possible if
    $q.c^a=0$ and $p.c^b=B-1$.  But then the definition of a $D$-path
    requires that $q.s^a=E$, and $p$ can execute the rule $R_P(1)$
    and thus
    cannot execute the rule $R_U$ in $\gamma^a\mapsto\gamma^b$.
  \end{itemize}
  The first part of the lemma follows.  Now if $p.s^b=C$, then $p$
  does not execute an error rule in $\gamma^a\mapsto\gamma^b$, and thus
  $p.c^b\geq p.c^a$, which completes the proof.
\end{proof}

\begin{lemma}\label{lem:D-pathGlop}
  Let $\gamma^a\mapsto\gamma^b$ be a step, let $P=p_0\cdots p_l$ be a
  decreasing path in~$\gamma^a$ such that
  \begin{itemize}
  \item apart from $p_l$ which satisfies $p_l.s^a=E$ and $p_l.s^b=C$,
    all the nodes of $P$ are in $C$ in both $\gamma^a$ and $\gamma^b$;
  \item in $\gamma^a$, $p_0\cdots p_{l-1}$ is gently decreasing.
  \end{itemize}
  Then $P$ is gently decreasing in $\gamma^b$.
\end{lemma}
\begin{proof}
  The assumptions imply that $p_l$ executes the rule $R_C$ in the step
  $\gamma^a\mapsto\gamma^b$.  Thus $p_l.c^b = p_l.c^a$.  
  
  We claim that, for any $0\leq i < l$, $p_i.c^b = p_i.c^a$.  Indeed,
  since $p_l.s^a=E$, by Lemma~\ref{lem:E-path}, $p_l$ is the first
  node of an $E$-path in $\gamma^a$ that we use to extend $P$ into
  a $D$-path $P'$. The claim then follows by
  Lemma~\ref{lem:D-path_time} applied to $P'$.
  
  The path $P$ is decreasing in $\gamma^a$, and in particular $p_{l-1}.c^a> p_l.c^a$.
  Moreover, $p_l$ executes the rule $R_C$, and thus we have $p_{l-1}.c^a = p_l.c^a+ 1$.
  As the beginning of the path is gently decreasing by hypothesis, $P$
  is gently decreasing in $\gamma^a$.  Finally, since the clock values
  of nodes of $P$ are the same in $\gamma^a$ and in $\gamma^b$, the
  lemma follows.
\end{proof}

\begin{lemma}\label{lem:D-pathTrans}
  Let $\gamma^a\mapsto\gamma^b$ be a step. Let $p$ be the first node
  of a $D$-path $P$ in~$\gamma^a$.  If at least one node of $P$ is in
  error in $\gamma^b$, then $p$ is the first node of a $D$-path in
  $\gamma^b$.
\end{lemma}
\begin{proof}
  Let $P=p_0\cdots p_l$ be a $D$-path in $\gamma^a$ and let $p=p_0$.
  Assume that $P$ contains at least one node in error in $\gamma^b$,
  and let $0\leq i\leq l$ be minimal such that $p_i.s^b=E$.

  Let $P'$ be the possibly empty path $p_0\cdots p_{i-1}$.  Since
  $p_i.s^b=E$, there exists an $E$-path $Q=p_iq_1\cdots q_h$ in
  $\gamma^b$, by Lemma~\ref{lem:E-path}.  We now claim that
  $P''=p_0\cdots p_iq_1\cdots q_h$ is a $D$-path in $\gamma^b$ whose
  first node is $p$.

  We first prove that $P''$ is decreasing.  Indeed, by
  Lemma~\ref{lem:D-path_time}, $p_i.c^b\leq p_i.c^a$ and, for
  $0\leq j < i$, $p_j.c^b = p_j.c^a$.  Since $P$ is decreasing in
  $\gamma^a$, so is $P'p_i$ in $\gamma^b$.  Now
  $p_iq_1\cdots q_h$ is an $E$-path in $\gamma^b$ and is thus also
  decreasing which implies that so is $P''$.

  To finish the proof, we must show that $P'$ is gently decreasing in
  $\gamma^b$. Let $P_c$ be the correct part of $P$ in $\gamma^a$.
  Since both $P'$ and $P_C$ are prefixes of $P$, we have 2 cases:
  \begin{itemize}
\item
  Assume that $P'$ is a prefix of $P_C$. Since $P_C$ is gently
    decreasing in $\gamma^a$, so is $P'$.
    And since all  nodes of
    $P'$ are still correct in $\gamma^b$, Lemma~\ref{lem:D-path_time} implies that $P'$ is
    gently decreasing in $\gamma^b$.
  \item 
  Otherwise, $P_C$ is a strict prefix of $P'$.  Since, in a
    $D$-path, at most one node can execute the rule $R_C$, we have
    $P'=P_C p_{i-1}$.  The fact that $P'$ is
    gently decreasing in $\gamma^b$ follows from
    Lemma~\ref{lem:D-pathGlop}.
  \end{itemize}
\end{proof}

\begin{lemma}\label{lem:languageInSegment}
  Let $\gamma^a\mapsto\gamma^b$ be a step, let $p$ be the first node
  of a $D$-path, and let $r$ be its root in $\gamma^a$.  If $r$ is still a
  root in $\gamma^b$, then $p$ is the first node of a $D$-path in
  $\gamma^b$.
\end{lemma}
\begin{proof}
  If $r$ executes the rule $R_C$ during $\gamma^a\mapsto\gamma^b$, then
  Lemma~\ref{lem:root_RC} implies that $r$ is not a root in
  $\gamma^b$, which is a contradiction.  Thus $r$ is in error in
  $\gamma^b$, and the lemma follows from Lemma~\ref{lem:D-pathTrans}.
\end{proof}

\begin{lemma}\label{lem:2theGround}
  Let $\gamma^a\mapsto\gamma^b$ be a step. Let $p$ be the first
  node of a $D$-path in~$\gamma^a$.  If no $D$-paths in $\gamma^b$
  contain $p$, then $p.c^b\leq -B+n$.
\end{lemma}
\begin{proof}
  Let $p$ be the first node of a $D$-path $P$, and let $r$ be the root
  of $P$ in~$\gamma^a$.

  We claim that, in $\gamma^b$, $P$ contains no nodes in error.
  Indeed, otherwise Lemma~\ref{lem:D-pathTrans} implies that $p$ is
  the first node of a $D$-path in $\gamma^b$, which is a
  contradiction.

  Since, in a $D$-path, at most one node can execute the rule $R_C$
  during a step, then in $\gamma^a$, all the nodes of $P$ but $r$ have
  status $C$.  We can thus apply Lemma~\ref{lem:D-pathGlop} and
  obtain that $P$ is gently decreasing in $\gamma^b$, and thus
  $p.c^b=length(P)+r.c^b$.  Since no nodes can appear twice in $P$,
  we have $p.c^b \leq r.c^b+n$.

  Now since $r.s^b=C$, $r$ executes the rule $R_C$ in
  $\gamma^a\mapsto\gamma^b$.  But then Lemma~\ref{lem:root_RC} implies
  that $r.c^a=-B$, and thus $r.c^b=-B$, and the lemma follows.
\end{proof}

 \subsection{Bounds on the clock values}
 
\begin{lemma}\label{lem:2D+1}
  If $i<j$ and $p$ satisfies $p.c^j> p.c^i+2D$, then for any $q$, there
  exists $i\leq h<j$ such that $q.c^h=p.c^i+D$ and $q$ executes the rule $R_U$ in the step $\gamma^{h} \mapsto \gamma^{h+1}$.
\end{lemma}
\begin{proof}
  First, notice that $p.c^i+2D < B-1$  by hypothesis.
  Then, we prove by induction on $d(q, p)$ that there exist
  $i\leq i'<j'\leq j$ such that $q.c^{i'}\leq p.c^i+d(q, p)$ and
  $p.c^j-d(q, p)\leq q.c^{j'}$.

  \begin{itemize}
  \item If $d(q, p)=0$, then $q=p$ and $i'=i$ and $j'=j$ do the trick.

  \item If $d(q, p)>0$ then let $q'\in N(q)$ be such that
    $d(q', p)=d(q, p)-1$.  By induction, there exists
    $i\leq i_1<j_1\leq j$ such that $q'.c^{i_1}\leq p.c^i+d(q, p)-1$
    and $p.c^j-d(q, p)+1\leq q'.c^{j_1}$.

    Now $q'.c^{j_1}-q'.c^{i_1}\geq p.c^j-p.c^i-2(d(q, p)-1)> 2D-2(d(q,
    p)-1)$. So, $q'.c^{j_1}-q'.c^{i_1}>2$.  Thus, there exists
    $i_1\leq i'<j_1$ such that $q'.c^{i'}=q'.c^{i_1}$ and $q'$
    executes the rule $R_U$ in $\gamma^{i'} \mapsto \gamma^{i'+1}$.  Since
    $q$ is a neighbor of $q'$, we have
    $q.c^{i'}\leq q'.c^{i'}+1=q'.c^{i_1}+1 \leq p.c^i+d(q, p)$.

    Now since $q'.c^{i'+1}+2\leq q'.c^{j_1}$, there exists
    $i'< j'<j_1$ such that  $q'$
    executes the rule $R_U$ in $\gamma^{j'} \mapsto \gamma^{j'+1}$ and
    $q'.c^{j'+1}=q'.c^{j_1}$.  Since $q$ is a neighbor of
    $q'$, we have
    $q.c^{j'}\geq q'.c^{j'}=q'.c^{j_1}-1\geq p.c^j-d(q, p)$, which
    finishes the proof of our induction.
  \end{itemize}

  Let $q$ be any node.  Let $i\leq i'<j'\leq j$ such that
  $q.c^{i'}\leq p.c^i+d(q, p)\leq p.c^i+D$ and
  $q.c^{j'}\geq p.c^j-d(q, p)\geq p.c^j-D> p.c^i+D$.  There exists
  $i'\leq h<j'$ such that $q.c^h=p.c^i+D$ and $q$ executes the rule $R_U$ in the step $\gamma^{h} \mapsto \gamma^{h+1}$.
\end{proof}

\begin{lemma}\label{lem:atMost2D}
  Suppose that $\gamma^h$ is not clean.  For any node $p$ and any
  $i<j\leq h$, $p.c^j-p.c^i\leq 2D$.
\end{lemma}
\begin{proof}
  Let $r$ be a root in $\gamma^h$ and let $i<j\leq h$.  By
  Lemma~\ref{lem:noRootCreation}, $r$ is a root in every configuration
  $\gamma^l$ with $l\leq h$, and since no roots can execute the rule
  $R_U$, the lemma follows from Lemma~\ref{lem:2D+1}.
\end{proof}

\subsection{Move complexity}

In this section, we analyze the move complexity of our algorithm. To
do so, we fix an execution $e=\gamma^0\gamma^1\cdots$ and study the
rules a given node executes in it.  Since these rules do not appear
explicitly in an execution, we propose to use a proxy for them.

A pair $(p, i)$ is a \emph{move} if $p$ executes a rule in
$\gamma^i\mapsto\gamma^{i+1}$.  This move is a \emph{$U$-move} if the
rule is $R_U$, a \emph{$C$-move} if the rule is $R_C$, a
\emph{$R$-move} if the rule is $R_R$, and a \emph{$P(i)$-move} if the
rule is $R_P(i)$.  Since a node $p$ executes at most one rule in a
given step, the number of steps in which a given node executes a rule
is the number of its moves.

Let $S_i$ be the set of roots in $\gamma^i$.
Lemma~\ref{lem:noRootCreation} states that for each $i>0$,
$S_{i}\subseteq S_{i-1}$.  Since $\gamma^0$ contains at most $n$
roots, there are $l\leq n$ steps $\gamma^{i-1}\mapsto\gamma^{i}$ for
which $S_{i}\subset S_{i-1}$. Let $r_1, r_2, \ldots, r_l$ be the
sequence of increasing indices such that
$\forall i \in [1,l], S_{r_i}\subset S_{r_i-1}$. This sequence gives
the following \emph{decomposition} of $e$ into segments.
\begin{itemize}
\item The \emph{first segment} is the sequence
  $\gamma^0\cdots\gamma^{r_1}$.
\item For $1<i\leq l$ the \emph{$i$-th segment} is the sequence
  $\gamma^{r_{i-1}}\cdots\gamma^{r_i}$.
\item The \emph{last segment} is the sequence $\gamma^{r_l}\cdots$.
\end{itemize}

A segment is said to be $clean$ if its first configuration  is clean.
If the first configuration of a segment has a root, then the segment is said to be $unclean$.
According to Lemma~\ref{lem:noRootCreation},
if the first configuration of a segment is clean then 
the other configurations of the execution are clean.  
So, there is
at most one clean segment, the last one, in any execution.

\paragraph{$R$-moves.}

\begin{lemma}\label{lem:nbRmoves}
  A node $p$ executes at most one $R$-move.
\end{lemma}
\begin{proof}
  Let $p$ be a node.  We have three cases.

  \begin{itemize}
  \item If $p$ executes no $R$-moves, it executes at most one $R$-move.
  \item If $p$ executes a $R$-move and no moves after the first
    $R$-move, then $p$ executes only one $R$-move.
  \item Otherwise, let $(p,i)$ be the first $R$-move (thus
    $p.c^{i+1}=-B$ and $p.s^{i+1}=E$), and let $(p,j)$ be the first
    move which follows.  Consequently, $(p,j)$ is necessarily a
    $C$-move. The result then follows from Lemmas~\ref{lem:root_RC}
    and~\ref{lem:noRootCreation}.
  \end{itemize}
\end{proof}

\paragraph{$U$-moves.} Note here that the predicate $P_\text{aux}$ can only prevent a node from executing
the rule $R_U$.  Hence, since we consider distributed unfair daemons,
an execution with any predicate $P_\text{aux}$ is a valid execution
with the predicate $P_\text{aux}=true$ while the configuration is not
clean.  We therefore consider in this part of the analysis that
$P_\text{aux}=true$.

\begin{lemma}\label{lem:errorInSegment}
Let $s$ be a segment.
  All  $U$-moves done by $p$ during $s$ are done consecutively before 
  the first error rule executed by $p$ during $s$ (if it exists).
\end{lemma}
\begin{proof}
  By definition of the rules $R_U$ and $R_C$, $U$-moves of $p$ are done
  consecutively before the first error rule executed by $p$.  According
  to Lemma~\ref{lem:E-path}, after $p$ executes an error rule, $p$ is
  the first node of an $E$-path, and thus of a $D$-path, by
  definition.  Lemma~\ref{lem:languageInSegment} implies that $p$
  remains in a $D$-path until the end of $s$. Hence, $p$ no more
  executes the rule $R_U$ in $s$, by Lemma~\ref{lem:D-path_time}, and we are
  done.
  \end{proof}
 
To compute the move complexity, we must, in particular, compute the total number
of moves in unclean segments.  By definition, the rules $R_R$,
$R_P$ and $R_C$ can only appear in unclean segments.
%
%
\begin{lemma}\label{lem:2D+1-RU}
Let $s$ be an unclean segment.
A node $p$ executes the rule $R_U$ at most $2D$ times during $s$.
\end{lemma}
\begin{proof}
  By definition of $s$, there is a node $r$ that is a root all along $s$.
  We now show, by induction on $d$, that every node $p$ at distance $d\leq D$ from $r$ executes at most $2d$ $U$-moves in $s$.
  \begin{description}
  \item[Base Case:] If $d = 0$, then $p = r$. Now, $r$ cannot execute a $U$-move during~$s$.
  \item[Induction Step:] Assume that $p$ is at distance $d > 0$ from $r$.
    Let $q \in N(p)$ such that $q$ is at distance $d-1$ from $r$.
    By Lemma~\ref{lem:errorInSegment}, if $p$, resp. $q$, changes its
    clock value during $s$, it does so by first executing a (possibly
    empty) sequence of $U$-moves, and then by executing a (possibly
    empty) sequence of error moves.  By induction hypothesis, $q$
    executes $x \leq 2(d-1)$ $U$-moves in $s$. To prove the induction
    step, it is sufficient to prove that $p$ does not execute more
    than $x+2$ steps during $s$.
    
    For the purpose of contradiction, assume that $p$ executes at
    least $x+3$ $U$-moves in $s$.  Let $c_p$ be the clock value of~$p$
    just before its first $U$-move in $s$.  There are $x+3$ integers $t_1 < t_2
    \cdots < t_{x+3}$ such that $(p,t_i)$ is a $U$-move in $s$ setting
    $p.c$ to the value $c_p +_B i$.  By definition of the rule~$R_U$,
    we must have $q.c^{t_i} \in \{ c_p +_B (i-1), c_p +_B i\}$.
    
    We claim that for any $1 \leq i \leq x+3$, node $q$ has executed
    at least $i-2$, resp. $i-1$, $U$-moves between the beginning of
    the segment and $\gamma^{t_i}$ when $q.c^{t_i} = c_p +_B (i-1)$,
    resp. $q.c^{t_i} = c_p +_B i$.  We prove this claim by induction
    on $i$. The base case $i=1$ is trivial.  Assume that the property
    holds for $i \geq 1$ and let us consider the different cases. If
    $q.c^{t_{i+1}} = q.c^{t_i}$, then $q.c^{t_i} = c_p +_B i$ and we
    immediately have the desired property by induction hypothesis.
    Otherwise, we have $q.c^{t_{i+1}} = q.c^{t_i} +_B j$, with $j$
    being $1$ or $2$.  Since $B \geq 4$, the value $q.c^{t_{i+1}}$ is
    either non-negative, or larger than $q.c^{t_i}$.  Since executing
    an error rule always decreases the clock value, and sets it to a
    negative value, $q$ cannot use any error rule to obtain for the
    first time the clock value $q.c^{t_{i+1}}$ from configuration
    $\gamma^{t_i}$. Therefore, $q$ must perform at least $j$ $U$-moves
    between $\gamma^{t_i}$ and $\gamma^{t_{i+1}}$.  Still by induction
    hypothesis, we thus obtain the desired property also in this case,
    which concludes the proof of the claim.  Using it with $i=x+3$
    allows us to obtain the expected contradiction, hence proving the
    overall induction step.
  \end{description}
  The lemma directly follows from the overall induction.
\end{proof}

\begin{lemma}\label{lem:nbUmoves}
  A node $p$ has at most $2Dn$ $U$-moves in the unclean
  segments.
\end{lemma}
\begin{proof}  
  By Lemma~\ref{lem:2D+1-RU}, $p$ executes the rule $R_U$ at most $2D$
  times in an unclean segment.  Since there are at most $n$ unclean
  segments (Lemma~\ref{lem:noRootCreation}), the lemma follows.
\end{proof}

\paragraph{$P$-moves with $B$.}
We bound the number of $P$ moves in 2 ways: using $B$, and without
using~$B$.

\begin{lemma}\label{lem:nbP-movesB}
  A node can have at most $nB$ $P$-moves.
\end{lemma}
\begin{proof}
  Let $p$ be a node.  In a clean segment, $p$ cannot execute a
  rule $R_P$.  In an unclean segment, by
  Lemma~\ref{lem:errorInSegment}, once $p$ executes a $P$-move, it
  cannot execute the rule $R_U$ anymore.  Each time $p$ executes
  the rule $R_P$, the variable $p.c$ decreases by at least one and
  takes a value in $[-B,0[$ .
      Hence, $p$ can only execute $B$ $P$-moves in an unclean
      segment. Since there are at most $n$ unclean segments
      (Lemma~\ref{lem:noRootCreation}), the lemma follows.
\end{proof}

\paragraph{$P$-moves without $B$.} We now need several definitions.

We say that a $P$-move $(p, t)$ \emph{causes} another $P$-move
$(p', t')$ if
\begin{itemize}
\item $p'\in N(p)$, $t'>t$,
\item for some $l$, $(p', t')$ is a $P(l)$-move and $(p, t)$ is a
  $P(l-1)$-move, and
\item for any $t<k<t'$, $(p, k)$ is not a move.
\end{itemize}

If a node $p$ is in error in some configuration $\gamma^i$, this often
happens because of some previous $P$-move $(p, t)$.  Moreover, what
allowed $(p, t)$ is some $q\in N(p)$ which is in error in
$\gamma^{t-1}$.  Finally, the reason why $q$ is in error in
$\gamma^{t-1}$ is because of some previous move and so on.  This
motivates the following definition: a \emph{causality chain} is a
sequence $C=(p_0, t_0)(p_1, t_1)\dots(p_l, t_l)$ such that
\begin{itemize}
\item for each $0\leq i<l$, $(p_i, t_i)$ causes $(p_{i+1}, t_{i+1})$;
\item no $(p, t)$ causes $(p_0, t_0)$.
\end{itemize}

By construction, any $P$-move is the last element of a causality chain
but the causality chain may not be unique.

We classify the $P$-move of $p$ in 3 types.
\begin{itemize}
\item $(p, i)$ is of Type 1 if there exists a $P$-move $(p, j)$ with
  $j>i$ such that $p.c^{i+1}=p.c^{j+1}$.
\item $(p, i)$ is of Type 2 otherwise.  And we subdivide Type 2
  $P$-moves in
  \begin{itemize}
  \item Type 2a.  if at least one causality chain
    $C=(p_0, t_0)\dots(p_l, t_l)$ ending in $(p, i)$ does not contain
    a repeated node.  More formally, for any $0\leq i<j\leq l$,
    $p_i\neq p_j$.
  \item Type 2b.  otherwise.
  \end{itemize}
\end{itemize}



Our goal is to separately bound the number of $P$-moves of each type
that a node can execute.

\begin{lemma}\label{lem:nbGlopType1}
  There are at most as many $P$-moves of type 1 as there are $U$-moves
  in the unclean segments.
\end{lemma}
\begin{proof}
  Suppose that $(p, i)$ and $(p, j)$ are both $P(l)$-moves with $i<j$.
  This means that $p.c^{i+1}=p.c^{j+1}=l$.  For $(p, j)$ to be
  possible, $p.c$ has to go from $l$ in $\gamma^{i+1}$ to being
  strictly greater than $l$ in $\gamma^j$.  This implies that there
  exists $i<k<j$ such that $(p, k)$ is a $U$-move with $p.c^k=l$.

  Thus, if we associate to each $(p, i)$ of type 1 the $U$-move
  $(p, j)$ such that $p.c^{i+1}=p.c^j$ with $j>i$ minimum, then no 2
  distinct $P$-moves correspond to the same $U$-move.  This implies
  that $p$ has at most as many $P$-moves of type 1 as it has $U$-moves
  in unclean segments.
\end{proof}

Remark that, by definition, two $P$-moves $(p, i)$ and $(p, j)$ of
type 2 are such that $p.c^{i+1}\neq p.c^{j+1}$.  To count the number
of $P$-moves $(p, i)$ of type 2, we thus count the number of values
that $p.c^{i+1}$ can take.

\begin{lemma}\label{lem:nbGlopType2a}
  A node $p$ can have at most $n(n+1)$  $P$-moves of  type 2a.
\end{lemma}

\begin{proof}
  Let $(p, i)$ be a $P$-move of type 2a, and let
  $C=(p_0, t_0)\dots(p_l, t_l)$ be a corresponding causality chain.
  We have
  \begin{itemize}
  \item $(p, i)=(p_l, t_l)$
  \item for any $0\leq i<j\leq l$, $p_i\neq p_j$.
  \end{itemize}
  Clearly, $l<n$ and $p_l.c^{t_l+1}=l+p_0.c^{t_0+1}$.  Let
  $r\in N(p_0)$ be such that $r.s^{t_0}=E$ and
  $p_0.c^{t_0+1}=r.c^{t_0}+1$.  Since no $P$-move causes $(p_0, t_0)$,
  two cases arise:
  \begin{itemize}
  \item the last move of $r$ before $t_0$ is an $R$-move in which case
    $r.c^{t_0}=-B$,
  \item $r$ executes no rule before $t_0$ in which case
    $r.c^{t_0}=r.c^0$.
  \end{itemize}
  Thus $p_0.c^{t_0+1}$ can take at most $n+1$ distinct values.
  The lemma now follows from the fact that $l$ can take at most
  $n$ distinct values.
\end{proof}

\begin{lemma}\label{lem:nbGlopType2b}
  A node $p$ can have at most $2n+2D$ type 2b $P$-moves.
\end{lemma}
\begin{proof}
  Let $(p, i)$ be a $P$-move of type 2b, and let
  $C=(p_0, t_0)\dots(p_l, t_l)$ be a causality chain such that
  $(p, i)=(p_l,t_l)$.

  By definition, there exists $0\leq i<j\leq l$ such that $p_i=p_j$.
  Choose such a $i_0=i$ and $j_0=j$ with $j_0$ maximum.  We thus have
  that for any $j_0\leq i<j\leq l$, $p_i\neq p_j$ and thus $l-j_0<n$.
  Let $q=p_{j_0}=p_{i_0}$.

  Now $p_l.c^{l+1}=q.c^{j_0+1}+(l-j_0)$.  To prove the lemma, it is
  thus enough to show that $q.c^{j_0+1}\leq n+2D$.

  We have that $q.s^{i_0+1}=E$, thus, by Lemma~\ref{lem:E-path}, $q$
  is the first node of an $E$-path, and thus of a $D$-path in
  $\gamma^{i_0+1}$.

  Since $q.c^{j_0+1}>q.c^{i_0+1}$, $q$ executes a $U$-move $(q, k)$ for
  $i_0<k< j_0$.  By Lemma~\ref{lem:D-path_time}, $q$ belongs to no
  $D$-path in $\gamma^k$.  There thus exists $i_0\leq k'<k$ such that
  $q$ belongs to a $D$-path in $\gamma^{k'}$ and to no $D$-path in
  $\gamma^{k'+1}$.  By Lemma~\ref{lem:2theGround}, $q.c^{k'+1}\leq n$.

  Since $q$ is in error in $\gamma^{j_0}$, by Lemma~\ref{lem:E-path},
  $q$ belongs to an $E$-path in $j_0$.  There thus exists a root $r$
  in $\gamma^{j_0}$.  By Lemma~\ref{lem:atMost2D},
  $q.c^{j_0+1}\leq n+2D$.  The lemma follows.
\end{proof}

Lemmas~\ref{lem:nbP-movesB}-\ref{lem:nbGlopType2b} directly
imply the following lemma.
\begin{lemma}\label{lem:nbPmoves}
  During an execution, there are at most $O(n^3)$
  $P$-moves.
\end{lemma}

\paragraph{$C$-moves.}

\begin{lemma}\label{lem:nbCmoves}
  During an execution, the number of $C$-moves is at most the number
  of $P$-moves plus $n$.
\end{lemma}
\begin{proof}
  Between two $C$-moves, a node $p$ must execute an error move.

  But since, after a $C$-move, $p$ is can no longer be a root (by
  Lemmas \ref{lem:noRootCreation} and~\ref{lem:root_RC}), $p$ cannot
  execute a $C$-move before an $R$-move. Thus $p$ can execute at most
  one more $C$-move than its number of $P$-moves.
\end{proof}

\paragraph{The move complexity theorem.} The following theorem is a direct corollary of
Lemmas~\ref{lem:nbRmoves},~\ref{lem:nbUmoves},~\ref{lem:nbPmoves}, and~\ref{lem:nbCmoves}.

\begin{theorem}\label{thm:stepComplexity}
  Our algorithm converges in  $O(\min(n^2B, n^3))$ moves.
\end{theorem}

\subsection{Round complexity}

Throughout this section, we consider an arbitrary execution
  $e = \gamma^0 \cdots$.  For all $i\geq 1$, we denote by
  $\gamma^{h_i}$ the last configuration of the $i^{th}$ round ({\em
    n.b.}, $e$ is finite, by Theorem~\ref{thm:stepComplexity}, so
  there is no infinite round in $e$ and from the last configuration of
  $e$, rounds are empty). We also let $\gamma^{h_0} = \gamma^0$.

 In the first $D+1$ rounds, nodes execute
error rules to ``correct'' the initial configuration. During the $D+1$
next rounds, all nodes go back to the correct state.  
The predicate $P_\text{aux}$ has no influence on results of this section as $R_U$ executions along $e$ do not impact our analysis.

\paragraph{The ``error broadcast phase''.}
\begin{lemma}\label{lem:penteDouce}
  For any $h\geq h_{D+1}$, in $\gamma^h$, for any root $r$, we have 
  $r.S=E$ and $p.c\leq -B+d(r,p)$, for any node $p$.
\end{lemma}

\begin{proof}
  If $\gamma^h$ contains no root, then the lemma holds.  Otherwise,
  let $r$ be any such root. By Lemma~\ref{lem:noRootCreation}, $r$ is
  also a root in all $\gamma^i$ with $i\leq h$.

  We first prove that $r.c^{h'} = -B$ and $r.s^{h'} = E$ for any $h'$
  with $h_1\leq h'\leq h$.  This claim will establish the first part
  of the lemma and the base case of the next induction.

  First, during the first round, while $r.c \neq -B$ or $r.s \neq E$,
  $r$ is enabled for $R_R$. Hence, by definition of a round and Rule
  $R_R$, there is a configuration in the first round where $r.c = -B$
  and $r.s = E$. From such a configuration, the next rule $r$ may
  execute is $R_C$. Now, by executing $R_C$, $r$ is not a root
  anymore, by Lemmas~\ref{lem:noRootCreation}-\ref{lem:root_RC}. So,
  $r$ cannot execute $R_C$ before the system reaches Configuration
  $\gamma^h$. Hence, for any $h'$ with $h_1\leq h'\leq h$, $r.c^{h'} =
  -B$ and $r.s^{h'} = E$.

  We now prove by induction on $j\geq 1$ that for all nodes $p$ such
  that $d(p, r)< j$, $p.c^{h'}\leq -B+d(r,p)$ with $h_j\leq h'\leq
  h$.

  If $j=1$, then $p=r$ and the base case is trivial from the previous
  claim.
  Suppose now that $j>1$.  Let $p$ be such that $d(r, p)=j$, and let
  $q\in N(p)$ be such that $d(r, q)=j-1$. By induction hypothesis, we
  have $q.c^{h'}\leq -B+j-1$ with $h_{j-1}\leq h'\leq h$.
  
  We first prove that there exists $h'$ such that $h_{j-1}\leq h'\leq
  h_j$ such that $p.c^{h'}\leq -B+j$.  To do so, assume, by the
  contradiction, that for every $h'$ with $h_{j-1}\leq h'\leq h_j$,
  $p.c^{h'}> -B+j$, which implies that $p.c^{h'}\geq q.c^{h'}+2$.
  From the previous claim, we also know that $p$ is not a root in any
  $\gamma^{h'}$.  Assume that $q.s^{h'}=C$ for some $h'$ with
  $h_{j-1}\leq h'\leq h_j$.  Then, $q$ is a root in $\gamma^{h'}$ and
  in $\gamma^0$, by Lemma~\ref{lem:noRootCreation}.  From the previous
  claim, we know that $q.s^{h''}=E$ for any $h''$ such that $h_1\leq
  h''\leq h$. Now, since $j-1 \geq 1$, we obtain a contradiction.
  Thus, $q.s^{h'}=E$ and $p.c^{h'}\geq q.c^{h'}+2$ for any $h'$ with
  $h_{j-1}\leq h'\leq h_j$.  Hence, $p$ is enabled for executing
  $R_P(i)$ with $i\leq q.c^{h'}+1\leq -B+j$ in every configuration
  $\gamma^{h'}$. By definition of a round and Rules $R_P(i)$, there
  exists a configuration $h''$ with $h_{j-1}\leq h''\leq h_j$, where
  $p.c^{h'} \leq -B+j$, a contradiction.

Finally, recall that $q.c^{h''}\leq -B+j-1$ for every $h''$ such that
$h'\leq h''\leq h_{D+1}$, by induction hypothesis.  So, $p.c^{h''}\leq
-B+j$ since $p$ cannot execute $R_U$. Hence, we are done with the
induction and the lemma holds.
\end{proof}

\begin{lemma}\label{lem:errorBroadcastPhase}
  For any $h\geq h_{D+1}$, $\gamma^h$ is almost clean.
\end{lemma}
\begin{proof}

   By Lemma~\ref{lem:almostCleanClosed}, we only need to show that
   $\gamma^{h_{D+1}}$ is almost clean. To do so and according to
   Lemma~\ref{lem:almostClean}, we now show that no node can execute
   an error rule in $\gamma^{h_{D+1}}$.
  The fact that no node can execute the rule
  $R_R$ in $h_{D+1}$ follows from Lemma~\ref{lem:penteDouce}.
Assume that in $\gamma^{h_{D+1}}$ a node $p$ verifies $errorPropag(p,
i)$.  There exists $q\in N(p)$ such that $q.c<p.c-1$ and $q.s=E$.  By
Lemma~\ref{lem:E-path}, there exists a $E$-path $P$ of length $l$ from
$q$ to a root $r$.  This path implies that $q.c\geq r.c+l\geq r.c+d(r,
q)$.  But then $p.c>q.c+1\geq r.c+d(r, q)+1\geq r.c+d(r, p)$, which
contradicts Lemma~\ref{lem:penteDouce}.  Hence, we conclude that no
node verifies $errorPropag(p, i)$ in $\gamma^{h_{D+1}}$, and we are
done.
\end{proof}


\paragraph{The ``error cleaning phase''.}

\begin{lemma}\label{lem:errorCleaningPhase}
  For any $h\geq h_{2D+2}$, $\gamma^h$ is clean.
\end{lemma}
\begin{proof}
   We prove by induction on
  $0\leq i\leq D+1$ that for any $j\geq h_{D+1+i}$, $\gamma^j$
  contains no $E$-path of length $>D-i$.
  According to Lemma~\ref{lem:errorBroadcastPhase}, 
  $\gamma^j$ is almost clean.

  Suppose that $i=0$.  In $\gamma^{h_{D+1}}$, a root is in a $E$-path
  (by definition of almost clean).  If $\gamma^{h_{D+1}}$ contains no
  $E$-path, then $\gamma^{h_{D+1}}$ is clean as it contains no root.
  Otherwise, let $P$ be a $E$-path.  Let $p$ and $r$ be its first and
  last node, and let $l$ be its length.  By definition of a $E$-path,
  $p.c-r.c\geq l\geq d(p, r)$.  By Lemma~\ref{lem:penteDouce},
  $p.c-r.c\leq d(p, r)$.  We thus have that $l=d(p, r)\leq D$.  The
  base case thus holds.

  Suppose that the hypothesis holds for $i\geq 0$. Again, if
  $\gamma^{h_{D+1+i}}$ contains no $E$-path, then it is clean.  
  Let $P$ be an $E$-path in $\gamma^{h_{D+1+i}}$.  
  Let $p$  be the first node of $P$.
  Since no nodes can execute an
  error rule during the round $D+1+i$,
   then $P$ is also an $E$-path in $\gamma^{h_{D+i}}$.
  Moreover, $R_C$ is not enabled on $p$ in $\gamma^{h_{D+i}}$;
  otherwise, $p$
  would have done a move during the round $D+i+1$, and $p$ would not been 
  in a  $E$-path in $\gamma^{h_{D+i+1}}$.

  There thus exists $q\in N(p)$ such that $q.c>p.c$ which is in error
  in $\gamma^{h_{D+i}}$.  
  The path $qP$  is a $E$-path in $\gamma^{h_{D+i}}$.  
  By induction, the
  length of $qP$ is at most $D-i$, and thus the length of $P$ is at
  most $D-(i+1)$.  The hypothesis thus holds for $i+1$.

 For $h\geq h_{2D+2}$, $\gamma^h$ contains
  no $E$-path, which implies that $\gamma^h$ is clean.
\end{proof}

\paragraph{The round complexity proof.} Lemmas~\ref{lem:errorBroadcastPhase} and~\ref{lem:errorCleaningPhase}
directly imply that
\begin{theorem}\label{thm:roundComplexity}
  Our algorithm converges in $2D+2$ rounds.
\end{theorem}

\section{Synchronizer}\label{sec:synchroniseur}

Using folklore ideas (see, {\em e.g.},~\cite{AwKuMaPaVa93}
and~\cite{EmKe21}), we can use our unison algorithm to simulate any
synchronous self-stabilizing algorithm in an asynchronous environment
under an unfair daemon. We now study such a simulation.

\subsection{Time definition}
In everyday life, we have a distinction between the value of a clock
(modulo 24 hours) and the time. Both are obviously linked.  We would
like to make a similar distinction here.

Let $e= \gamma^0\cdots$ be an execution such that $\gamma^0$ is clean,
and so almost clean too.  According to Lemma~\ref{lem:colorValue},
there exists $c_{\min}\in [-B, B[$ and $\Delta\leq D$ such that
    $\{p.c^0\mid p\in V\}=\{c_{\min}+_B i\mid 0\leq i\leq \Delta\}$.
    The \emph{birth time} of $p$ is $time^0(p)=i-\Delta$ where $i$
    satisfies $p.c^0=c_{\min}+_B i$.  Moreover,
    $time^{j+1}(p):=time^j(p)+1$ whenever $p$ executes the rule $R_U$
    in $\gamma^j\mapsto\gamma^{j+1}$ (otherwise,
    $time^{j+1}(p):=time^j(p)$).

An important remark is that if $time(p)=time(q)$, then $p.c=q.c$.
Moreover, $unisonMove(p)$ is true if and only if $time(p)$ is a local
minimum.  Note that the birth time of a node is in $[-D,0]$.  

\subsection{The algorithm}

We consider a synchronous self-stabilizing algorithm $\origAlg$ which
runs in a variant of the atomic-state model which is at least as
expressive as the model of our unison algorithm.  This means that we
should be able to encode the macros \texttt{Macro1} and
\texttt{Macro2} (defined page~\pageref{macro}) in the model of
$\origAlg$. In the following, we denote by $T$ the stabilization time
of $\origAlg$ (in synchronous settings) and by $Trans(\origAlg)$ the
simulation of $\origAlg$ using our unison algorithm.

The basic idea of the simulation is that the execution of $\origAlg$
is driven by the unison algorithm. To that goal, each node $p$ stores
its last two states in $\origAlg$ using two
additional variables: $p.old$ and $p.curr$. Once the unison algorithm
has stabilized, if $p$ is a local minimum (w.r.t. the time of the
unison) and is about to increase its clock (by performing Rule $R_U$),
it computes its next state $\widehat{\origAlg}(p)$ in $\origAlg$. It does so by selecting for each neighbor $q$ the variable
$q.curr$ if $p.c=q.c$, and $q.old$ otherwise ({\em i.e.}, when
$q.c=p.c+_B 1$).

We thus modify the rule $R_U$ in the following way:
\begin{eqnarray*}
  R_U : unisonMove(p)&\wedge& (P_\text{aux}(p)\vee \exists
                                q\in N(p), q.c=p.c+_B 1) \\
                       &\longrightarrow&p.old:=p.curr;\\
                                &&p.curr:=\widehat{\origAlg}(p);\\
                                &&p.c:=p.c+_B1 
\end{eqnarray*}

Let us consider the execution after the unison has stabilized ({\em i.e.}, the
suffix of the execution starting from the first clean configuration).  The
time of each node is thus defined.  If $time(p)^i=t$, then we set
$\texttt{st}_p^t=p.curr$.  Since the state of $p$ changes if and only
if its time does, this is well defined.  For any positive $t$,
we can then define the configuration $\eta^t$ of
$\origAlg$ in which the state of each node $p$ is $\texttt{st}_p^t$.
The folklore claim is that the sequence $\eta^0\cdots$ is a
synchronous execution of $\origAlg$. 



When $P_\text{aux}(p)$ is always $true$, the clocks of the unison constantly
change. Therefore, even if $\origAlg$ is silent, its simulation is
not. In order to obtain a silent simulation in such a case, we instantiate
the predicate $P_\text{aux}(p)$ such that $p$ increments its clock (and thus
performs a simulation step) only if the simulation step makes its
state change.
More precisely, we define two possible predicates as follows:
\begin{eqnarray*}
  P_\text{greedy}(p) &=& true\\
  P_\text{lazy}(p) &=& p \text{ is enabled in $\origAlg$}
\end{eqnarray*}
and say that our synchronizer runs in \emph{greedy mode} if
$P_\text{aux}=P_\text{greedy}$, and that it runs in \emph{lazy mode} if
$P_\text{aux}=P_\text{lazy}$.

\subsection{Complexity analysis}

\paragraph{Greedy mode.} In greedy mode, Lemma~\ref{lem:vivant}  implies that the algorithm
is never silent.  Nevertheless, it is easy to see that, once unison
has been reached, all nodes with minimum time can be activated.  And
since nodes cannot be deactivated unless executing $R_U$, after one
round, the minimum time of a node has increased by at least one.
Thus, after $O(D)$ rounds (and $O(\min(n^2B, n^3))$ steps), each round
of $Trans(\origAlg)$ simulates at least one round of $\origAlg$.

\paragraph{Lazy mode.}

\begin{lemma}
  In lazy mode, the maximum time of each node is at most $T$.
\end{lemma}
\begin{proof}
  Let $T_{\eta^0}\leq T$ be the number of rounds that $\origAlg$ takes
  to be silent from the clean configuration $\eta^0$.  Let $e=
  \eta^0\cdots$ be an execution starting from $\eta^0$.  We claim that
  no node $p$ has the time $T_{\eta^0}+1$ along $e$.  Let $p$ be any
  node.  By time definition, $time^0(p)\leq 0$.

  Suppose that $\eta^i\mapsto\eta^{i+1}$ is 
  such that no
  node has a time greater than $T_{\eta^0}$ in $\eta^i$.
  If $time^i(p)\leq T_{\eta^0}-1$, then
  $time^{i+1}(p)\leq T_{\eta^0}$.  
  Otherwise, $time^i(p)=T_{\eta^0}$
  and no neighbor $q$ of $p$ is such that $q.c=p.c+_B 1$.  
  Moreover,
  by definition of $T_{\eta^0}$, 
  $p$ is not enabled in $\origAlg$.
  So $p$ cannot execute the rule $R_U$ 
  and $time^{i+1}(p)=T_{\eta^0}$ in $\eta^{i+1}$.  
\end{proof}

\begin{lemma}\label{lem:lazy_silent}
  If $\origAlg$ reaches a terminal configuration in at most $T$
  synchronous rounds, then $Trans(\origAlg)$ reaches a terminal
  configuration in at most $nT +nD$ moves from a clean configuration.
\end{lemma}
\begin{proof}
Let $e= \gamma^0\cdots$ be
an execution starting from a clean configuration. 
The birth time of a  node $p$ is in $[-D,0]$.
No node has a time greater than $T$ along $e$.
So a node executes the rule $R_U$ at most $T+D$ times along $e$.
\end{proof}

Because of the previous lemma, in lazy mode, if $\origAlg$ is silent,
then all executions of $Trans(\origAlg)$ are finite.
The round analysis is a bit more involved.  We split the analysis in
two parts: the number of rounds so that all nodes have positive time, and the
additional number of rounds to reach silence.

In the following, we consider a (finite) execution $e=\gamma^0\cdots
\gamma^f$ be an execution where $\gamma^0$ clean. As previously, we
denote by $\gamma^{h_i}$ the last configuration of the $i^{th}$ round
of $e$, for any $i\geq 1$, and we let $\gamma^{h_0} = \gamma^0$.

\begin{lemma}
\label{lem:2Drounds}
  In lazy mode, the time of all nodes is positive after at most $2D$
  rounds.
\end{lemma}
\begin{proof}
  Let $s$ be a node with birth time zero.  Let $\lambda(p, i):=2i+D+d(p, s)$.
  We prove by induction on $0\leq j\leq 2D$ that if $i\leq 0$ and
  $\lambda(p, i)\leq j$, then $time^{h_j}(p)\geq i$.

  Suppose that $j=0$. If $\lambda(p, i) \leq 0$, then $i \leq -d(p,
  s)$. As $\gamma^0$ is clean, we have $time^{h_0}(p)time^0(p)\geq
  -d(p, s) \geq i$ .  So, the base case holds.

  Suppose that $j>0$.  Let $(p, i)$ be such that $\lambda(p, i)=j$.
  For any $q\in N[p]$,
  $\lambda(p, i)-\lambda(q, i-1)=2+d(p, s)-d(q, s)>0$.  So
  $\lambda(q, i-1)<j$ and, by induction hypothesis,
  $time^{h_{j-1}}(q)\geq i-1$.

  Two cases now arise:
  \begin{itemize}
  \item If $time^{h_{j-1}}(p)\geq i$, then we are done.
  \item If $time^{h_{j-1}}(p)= i-1$, then $p\neq s$ (recall that
    $time^0(s)=0$, and so $time^{h_{j-1}}(s)\geq 0 \geq i$).  Then let
    $q\in N(p)$ be such that $d(q, s)<d(p, s)$.  We have $\lambda(q,
    i)<j$, and thus, by induction hypothesis, $time^{h_{j-1}} (q)\geq
    i$.  This implies that $p$ can execute the rule $R_U$ in
    $\gamma^{h_{j-1}}$, and thus will have done at last at
    $\gamma^{h_j}$.
  \end{itemize}
\end{proof}


We now focus on the part of the execution $e$ in which all nodes have a
positive time.  We denote by $e' = \rho^0\cdots\rho^f$ this part.  From now on,
 we
denote by $\rho^{r_i}$ the last configuration of the $i^{th}$ round
in $e'$, for any $i\geq 1$, and we let $\rho^{r_0} = \rho^0$.

We
do not know how these times evolve during the rounds of $e'$, 
nevertheless we know that if no
nodes have time $l+1$ in $\rho^i$ but $time^{i+1}(p)=l+1$, then
$p$ is enabled in $\origAlg$ at $\eta^l$.  We thus say that any such a node
$p$ \emph{may start
  time} $l+1$.

If  no node are enabled for $\origAlg$ in $\eta^i$, then in 
in $\eta^{i+j}$ with $ j >0$, no node are enabled in $\origAlg$.  
Therefore, if $p$ may start time  $i+1$, then either $i=0$
or there exists $q\in N[p]$ which may start time $i$.

This motivates the following definition. A \emph{starting sequence}
for $\rho^f$ is a sequence of nodes $s_1s_2\cdots s_H$ such that
each $s_i$ starts time $i$, and $s_{i-1}\in N[s_i]$ if $i>1$.  Note
that if $\rho^0$ contains no node which may start time $1$, then the
algorithm is already silent. Otherwise, $\rho^f$ must contain a
starting sequence.

\begin{lemma}\label{lem:roundAlgoPhaseLazy1}
$e'$ reaches a terminal configuration in at most $D+3T-2$ rounds
  in lazy mode.
\end{lemma}
\begin{proof}
  If $\rho^f$ contains no starting sequence, then
  $\rho^f=\rho^0$, and $e'$ indeed reaches a terminal configuration
  in at most $(D+3T-2)$ rounds.
  
  Assume now that $\rho^f$ contains the starting sequence $s_1\cdots
  s_T$.  We also let $s_i = s_1$, for any $i < 1$.  For any node $p$
  and $0\leq i\leq T$, we let $\lambda(p, i)=3i-2+d(p, s_i)$.  The
  lemma is a direct consequence of the following induction.

  We now prove by induction on $0\leq j\leq 3T+D-2$ that for every $p$
  and $i$ such that $\lambda(p, i)\leq j$, we have
  $time^{r_j}(p)\geq i$.

  If $j=0$, then $i=0$ and the result is clear.

  Suppose that $j> 0$.  If no $(p, i)$ such that $\lambda(p, i)=j$
  exists, then we are done.  Otherwise, let $(p, i)$ be such a pair.
  For any $q\in N[p]$,
  $\lambda(p, i)-\lambda(q, i-1)=3+d(p, s_i)-d(q, s_{i-1})$, and thus
  $\lambda(p, i)-\lambda(q, i-1)\geq 3-|d(p, s_i)-d(p,
  s_{i-1})|-|(d(p, s_{i-1})-d(q, s_{i-1})|$.  Now
  $|d(p, s_i)-d(p, s_{i-1})|\leq 1$ because $s_{i-1}\in N[s_i]$, and
  $|(d(p, s_{i-1})-d(q, s_{i-1})|\leq 1$ because $q\in N[p]$.  We thus
  have $\lambda(q, i-1)<j$.  By induction hypothesis, for any
  $q\in N[p]$, $time^{r_{j-1}}(q)\geq i-1$.

  Three cases now arise:
  \begin{itemize}
  \item If $time^{r_{j-1}}(p)\geq i$, then we are done.
  \item If $time^{r_{j-1}}(p)=i-1$ and $p=s_i$.  Then, since $s_i$
    may start time $i$, $p$ can execute the rule $R_U$ in
    $\rho^{r_{j-1}}$, and thus will have done at last at
    $\rho^{r_j}$.

  \item If $time^{r_{j-1}}(p)=i-1$ and $p\neq s_i$.  Then, let
    $q\in N(p)$ be such that $d(q, s_i)<d(p, s_i)$.  We have
    $\lambda(q, i)<j$, and thus, by induction hypothesis,
    $time(q)\geq i$ in $\rho^{r_{j-1}}$.  This implies that $p$ can
    execute the rule $R_U$ in $\rho^{r_{j-1}}$, and thus will have
    done at last at $\rho^{r_j}$.
  \end{itemize}
  Since $\lambda(p, i)\leq 3T+D-2$, the lemma follows.
\end{proof}

By Theorems~\ref{thm:stepComplexity} and~\ref{thm:roundComplexity} and Lemmas~\ref{lem:2Drounds}-\ref{lem:roundAlgoPhaseLazy1}, follows.

\begin{theorem}\label{them:lazy_silent}
Assumes that $\origAlg$ reaches a terminal configuration in at most
$T$ rounds and requires $O(M)$ bits per node. In lazy mode,
$Trans(\origAlg)$ reaches a terminal configuration in $O(\min(n^2B,
n^3))+nT$ moves and at most $5D+3T$ rounds. Moreover,
$Trans(\origAlg)$ requires $O(M+\log(B))$ bits per node.
\end{theorem}

\bibliographystyle{alphaurl}
\bibliography{biblio}
\end{document}